%% file: main.tex
\documentclass[sigconf]{acmart}
\usepackage{array}
\usepackage{balance}
\usepackage{appendix}
\newtheorem{definition}{Definition}

\newtheorem{lemma}{Lemma}
\input{macro_notations}
\input{macro_notations_extended}
\AtBeginDocument{%
  \providecommand\BibTeX{{%
    \normalfont B\kern-0.5em{\scshape i\kern-0.25em b}\kern-0.8em\TeX}}}

\renewcommand\footnotetextcopyrightpermission[1]{} % Removes footnote with conference information in first column
\begin{document}

\title{Sponsored Question Answering}

% \settopmatter{printacmref=true, printfolios=false}
% \fancyhead{}
\author{Tommy Mordo}
\affiliation{%
  \institution{Technion}
  \city{Haifa}
  \country{Israel}
}
\email{tommymordo@technion.ac.il}

\author{Moshe Tennenholtz}
\affiliation{%
  \institution{Technion}
  \city{Haifa}
  \country{Israel}
}
\email{moshet@technion.ac.il}

\author{Oren Kurland}
\affiliation{%
  \institution{Technion}
  \city{Haifa}
  \country{Israel}
}
\email{kurland@technion.ac.il}
\fancyhead{}
\renewcommand{\shortauthors}{Tommy Mordo, Moshe Tennenholtz, \& Oren Kurland}

\input{abstract}
%\input{category}
  % \includegraphics[width=\textwidth]{sampleteaser}
%  \includegraphics[width=\textwidth]{figs/sponsored content.png}
%  \caption{Sponsored Content in LLM-based Question Answering system}
\maketitle

\input{introduction}

\input{related}

%\input{preliminaries}
\input{model}
\input{conc}

\input{ack}
%\balance
%\balance
\input{appendix}
\bibliographystyle{ACM-Reference-Format}
\bibliography{bib_sponsored}
\end{document}

%% file: macro_notations.tex
\newcommand{\myparagraph}[1]{\vspace{0.6\baselineskip}\noindent{\textbf{#1}}.~}
\newcommand{\tf}[2]{tf(#1,#2)}

\newcommand{\simFnSymbol}{S}
\newcommand{\simFn}[2]{\simFnSymbol(#1,#2)}

\newcommand{\lm}{\theta}
\newcommand{\userUtil}[2]{U(#1,#2)}
\newcommand{\adUtil}[1]{U^{a}_{#1}}
\newcommand{\omt}[1]{}
\newcommand{\question}{q}
\newcommand{\numAds}{n}
\newcommand{\winner}{i^{w}}

\newcommand{\second}{i^{snd}}

\newcommand{\adi}{d_i}

\newcommand{\SCi}{d_{i}^{s}}
\newcommand{\SCj}{d_{j}^{s}}
\newcommand{\SCis}{d_{i^{w}}^{s}}
\newcommand{\SCit}{d_{\second}^{s}}
\newcommand{\pure}{d_{o}}
\newcommand{\voc}{V}
\newcommand{\textk}{V^k}

\newcommand{\inN}{i \in N}

\newcommand{\valuei}{v_i(\SCi)}
\newcommand{\valuej}{v_j(\SCj)}
\newcommand{\valueis}{v_{\winner}(\SCis)}
\newcommand{\valueit}{v_{\second}(\SCit)}

\newcommand{\bidi}{b_i(\SCi)}
\newcommand{\bidis}{b_{i^{w}}(\SCis)}
\newcommand{\bidit}{b_{\second}(\SCit)}
\newcommand{\bidj}{b_{j}(\SCj)}
\newcommand{\payi}{p_i(\SCi)}
\newcommand{\payis}{p_{i^{w}}(\SCis)}

\newcommand{\simdi}{S(\SCi, \adi)}
\newcommand{\simdo}{S(\SCi, \pure)}

\newcommand{\setup}{QAS := (\question,\pure, N, \{\adi\}_{i\in N}, \{\SCi\}_{i\in N}, \{\valuei\}_{i\in N})}
\newcommand{\VP}{PV}

%% file: macro_notations_extended.tex
\newcommand{\ada}{d_{1}}
\newcommand{\adb}{d_{2}}

\newcommand{\SCa}{d_{1}^{s}}
\newcommand{\SCb}{d_{2}^{s}}
\newcommand{\valuea}{v_{1}(\SCa)}
\newcommand{\valueb}{v_{2}(\SCb)}

\newcommand{\un}[2]{p_{#1}(#2)}

%% file: abstract.tex
\begin{abstract}
The potential move from search to question answering (QA) ignited the
question of how should the move from sponsored search to sponsored QA
look like. We present the first formal analysis of a sponsored QA
platform. The platform fuses an organic answer to a question with an ad
to produce a so called {\em sponsored answer}. Advertisers then bid on
their sponsored answers. Inspired by Generalized Second Price Auctions
(GSPs), the QA platform selects the winning advertiser, sets the payment she pays, and shows the user the sponsored answer. We prove an array of
results. For example, advertisers are incentivized to be truthful in
their bids; i.e., set them to their true value of the sponsored answer. The resultant setting is stable with properties of VCG auctions.
\end{abstract}

%% file: introduction.tex
\section{Introduction}
In sponsored search \cite{Jansern+Nullen:08a}, advertisements
(henceforth referred to as ads) are ranked in response to a query. A
click on an ad usually leads to a landing page on the Web. In contrast
to organic search, the ranking is not based on relevance estimation
but rather on monetization criteria. For example, by the foundational
General Second Price (GSP) Auction
\cite{edelman_internet_2007,varian_position_2007}, ads are ranked by
bids posted by the advertisers\footnote{Advertisers bid on keywords. If a keyword appears in a query then the ad is candidate for ranking.} where the payment of an advertiser is the bid posted by
the advertiser ranked just below her. More evolved versions include
measures of ad quality (e.g., predicted click through rate) in the
ranking criterion \cite{varian_position_2007,edelman_internet_2007}.

Recently, there is an on-going discussion, and some implementations,
of using question answering systems as potential alternatives to
search engines \cite{metzler_rethinking_2021,coldewey,you}. This is
due to the dramatic progress with large language models (LLMs)
\cite{zhao_survey_2023}. An emerging question is then how sponsored
search will evolve with the transition from search to question
answering. A few recent suggestions include showing ads during a QA
conversation (chat), providing the ad as an answer in a QA session, and
integrating an answer with an ad
\cite{coldewey,noauthor_new_nodate,feizi_online_2024,zelch_user_2024}.

However, to the best of our knowledge, there is still no theoretical
treatment of potential mechanisms for sponsored QA. In this paper, we
make a first step towards this end. We present a formal auction
mechanism for a sponsored QA platform. The use of ads in the QA
platform we study is inspired by Feizi's et
al. \cite{feizi_online_2024} conceptual proposal, by Zelch et al.'s study \cite{zelch_user_2024} and by Microsoft's
alleged approach \cite{perloff_microsoft_2023} to fuse (integrate) answers with ads. 
%with an ad to produce what we refer to as {\em sponsored answer}.
For example, suppose the 
question {\em ``Which book is used in most CS programs in the U.S. to teach Python?''} has an answer {\em ``Book
  X''} generated by some LLM. Suppose there is an ad about a specific bookstore Y: {\em ``Bookstore Y has all the programming books one might need to thoroughly learn how to program. For example, we highly recommend book X which consistently receives excellent reviews''}. The result of fusing the answer and the ad can be: {\em ``In bookstore Y you can find book X which is used by most CS programs in the U.S. to teach Python. The book consistently receives excellent reviews.''}

Our suggested QA platform operates as follows. Given a user question, the platform generates an {\em
  organic answer} which is independent of monetization
considerations, e.g., using an LLM. The platform fuses the answer with each of the
candidate ads to produce so called {\em sponsored answers}. Then, each
advertiser places a bid for her sponsored answer. Inspired by the
GSP auctions used in sponsored search
\cite{varian_position_2007}, we devise a criterion for selecting the
sponsored answer and the payment the selected advertiser should
pay. We prove that a {\em dominant strategy} of advertisers is to
place a bid equal to the value of the sponsored answer for them. This
result which is a property of VCG auctions \cite{MasColell+al:95a} is 
important since the auction setting is in a stable (equilibrium)
state: advertisers are incentivized to be truthful in their bids
rather than engage in {\em shading} \cite{Nautz+al:97a} --- i.e., continuously changing bids to improve gains.

%In contrast, in GSPs advertisers are not incentivized to be truthful which leads to non-stable settings where {\em shading}
%\cite{Nautz+al:97a} is quite common: advertisers continuously change their bids so as to improve gains.

We further define a notion of {\em social welfare} which is an
aggregate of the value of a sponsored answer to the advertiser and the
user utility attained by this answer. Then, in the spirit of VCG
auctions \cite{MasColell+al:95a}, we show that under the dominant
strategy, the utility of the winning advertiser (defined as the
difference between the value and the payment) is the difference
between the social welfare given her sponsored answer and the social
welfare given the sponsored answer of the second best advertiser whose
bid is used to determine the winner's payment.

Our formal analysis is agnostic to the fusion approach employed to
integrate the organic answer and an ad. We discuss a few potential
fusion methods using large language models (LLMs). However, formally
coupling the result of LLM-based fusion and the auction mechanism is
extremely difficult to impossible. Hence, to provide an end-to-end formal
analysis of the process that starts with fusion and continues with the
auction, we use unigram language models. We formally show that the
advertiser with the highest bid is not necessarily the winner and that
the winning sponsored result does not necessarily maximize user
utility. The latter result can be attributed to the fact that we
devise the platform to maximize both users' and advertisers'
satisfaction.

To summarize, our main contributions are (i) presenting the first
--- to the best of our knowledge --- formal proposal of a sponsored QA
system which is based on an auction mechanism, and (ii) proving a few important
theoretical results about the platform.

%% file: related.tex
\section{Related Work}
We assume that the QA platform generates an organic answer to a
question, e.g., using a large language model (LLM). Indeed, question answering is a
prominent application of using large language models
\cite{zhao_survey_2023}. We note that answer generation is not the focus of our
work here. Any LLM can be used to generate the organic answer in our suggested platform.

Sponsored search is central to the monetization of search engines
\cite{arnau_digital_2023}. In sponsored search, ads are ranked in
response to a query. In contrast, following recent proposals \cite{perloff_microsoft_2023,feizi_online_2024,zelch_user_2024}, we assume that an ad is fused with an organic answer to a question. 

Two prevalent auctions on the Web are VCG \cite{MasColell+al:95a} and
Generalized Second Price (GSP) auctions
\cite{edelman_internet_2007,varian_position_2007}. These auctions are
equivalent when a single advertisement is selected
\cite{vickrey_counterspeculation_1961,edelman_internet_2007}. In this
case, the highest bidding advertiser wins, and the payment is the
minimum bid necessary to secure the win, specifically, the bid of the
advertiser with the second highest bid. Our suggested auction is adaptation of a GSP auction to the sponsored QA setting. Since our QA platform strives
to satisfy both users and advertisers, the winning advertiser is not
necessarily the one who posted the highest bid as we show.

%In contrast to GSP which is not necessarily
%incentive compatible, we prove that our suggested auction
%mechanism is incentive compatible. That is, a dominant strategy of
%advertisers is to bid the true value of their sponsored answer. In
%other words, it is the incentive of advertisers to be truthful in
%their bids. This is important as the setting reaches a stable state (equilibrium) rather than faces shading \cite{Nautz+al:97a} where advertisers consistently adapt their bids so as to maximize gain.
%We note that VCG auctions are also incentive compatible. 

%We show that under certain
%conditions, the winner of the auction in our proposed QA platform is
%not the one with the highest bid. 

There is work on content generation using multiple LLMs where an
incentive-compatible auction\footnote{In these auctions, bidders are
  incentivized to bid their true value.} is held for {\em each}
generated token \cite{duetting_mechanism_2023}. Our platform is
fundamentally different as bids are posted for a given question and a
sponsored answer created from an ad and an organic
answer. Furthermore, in contrast to our work, there is strong coupling
between the content generation approach and the auction
\cite{duetting_mechanism_2023}. Our platform is agnostic to the QA
algorithm used to generate an organic answer and the approach used to
fuse it with an ad. Furthermore, in contrast to
Duetting et al. \cite{duetting_mechanism_2023}, we define and analyze social welfare.

Dubei et al. \cite{Dubei+al:24a} devise an auction mechanism where ads
of advertisers are summarized using an LLM. The advertisers
essentially bid on the prominence of their ad in the summary. In
contrast, in our QA platform a single sponsored answer created from an
ad and an organic answer is presented to the user. Hence, the
auction mechanisms are completely different and so is the
corresponding analysis.

Recently, Feizi et al. \cite{feizi_online_2024} proposed a conceptual
framework for online advertising in LLM-based QA systems. They discuss
important considerations including privacy, latency, reliability, users'
and advertisers satisfaction. However, no formal/algorithmic framework
is proposed in their work. We adopt Feizi et al.'s \cite{feizi_online_2024} proposal
of fusing an organic answer with an ad. This approach was also studied by Zelch et al. \cite{zelch_user_2024} and is also allegedly the
state-of-affairs in some implementations \cite{perloff_microsoft_2023}.
Our focus in this paper is on devising an auction mechanism and analyzing it using several perspectives. The important issues discussed by Feizi et al. \cite{feizi_online_2024} are outside the scope of this paper, but certainly deserve in-depth future exploration. A case in point, the reliability of sponsored answers generated from organic answers and ads is highly important.

Zelch et al. \cite{zelch_user_2024} performed a user study of a
setting where an organic answer is fused with an ad using large
language models. They analyzed the user judgements of the quality and
relevance of the resultant sponsored answers and the users' ability to
detect sponsored content. We focus on a theoretical analysis of the
auction mechanism. The quality and/or relevance of the sponsored answers obviously plays a major role in the user utility function which is part of our analysis.

Finally, we note that large language models were also recently used to detect automatically generated native advertisements
\cite{schmidt_detecting_2024}.

%such as reliability of sponsored answers generated from 

%However, the
%article does not delve into the analysis of the mechanism of selecting
%the advertisement nor the advertiser strategies. Furthermore, it does
%not address methods for generating an answer integrated with an ad,
%what we call later 'sponsored content'.

%% file: model.tex
\section{Sponsored Question Answering Platform}
\label{sec:model}
We now turn to present a question answering (QA) platform which 
responds to users' questions with answers which include sponsored information. 

%This is, for example, the state-of-affairs in commercial search engines (e.g., Google) that present both organic results and sponsored results in response to queries.

Suppose that a user posts a question $\question$ and the QA platform
generates an {\em organic} (textual) answer $\pure$ which is not affected by
monetization considerations. There are $\numAds$ advertisers (content
providers) interested in presenting their {\em textual} advertisements
(ads in short) in response to $\question$.\footnote{The advertisements
  can include images and video captions, but henceforth we focus on
  the textual part,} We assume that each advertiser $i$ ($\in
N:=\{1,2,...,n\}$) has a single ad $\adi$ she wants to present. The
platform fuses for every advertiser $i$ her ad ($\adi$) with the
organic answer $\pure$ to produce a {\em sponsored answer} $\SCi$
which will be shown to the user in case advertiser $i$ is
selected. The fusion process can be advertiser-specific and be based
on commercial terms between the platform and the advertiser. A case in
point, the relative emphasis on the organic answer versus that on the ad can be the result of a commercial agreement.  The platform then runs an auction in which the advertisers
bid on the corresponding sponsored answers.
The platform uses the
bids with additional information to determine the winner of the
auction and her payment.

%We note that our analysis of the auction is agnostic to the fusion
%process which we treat as a black box. 

Since auctions are essentially games composed of players (bidders) and
their strategies (bids) we start by describing in Section \ref{sec:gt}
some basic game theory concepts. In Section \ref{sec:auctions} we formalize
the auction mechanism and present some theoretical results. The
formalism and results are agnostic to the actual fusion approach
employed to fuse an organic answer and an ad. In Section
\ref{sec:unigram} we discuss two options of fusion using large
language models, and in addition provide a theoretical analysis of
the results of using unigram language models for fusion.

\input{preliminaries}

\subsection{Auctions for Sponsored Answers}
\label{sec:auctions}
We next turn to formalize the auction employed by the QA
platform. Recall that prior to the auction, each ad $\adi$ is fused
with the organic answer $\pure$ to yield a sponsored answer
$\SCi$. Advertiser $i$ then places a bid for $\SCi$ to be shown. The
platform selects the sponsored answer to show and determines the
payment of the winning advertiser. We note that fusing organic (non strategic) content and
strategic content (ads) is a
fundamental difference with work on sponsored search where only strategic content is used \cite{varian_position_2007}. Accordingly, the auction setting is significantly different.

\omt{
We define $n$ fixed possible sponsored contents, one for each
advertiser as $\{\SCi\}_{i \in N}, \forall i\in N : \SCi
\in\textk$. those sponsored contents can be determined through prior
negotiations between the platform and the advertisers. Later on, we
will demonstrate several possible processes for generating those
sponsored contents.

We begin by outlining the setup of sponsored content in QA systems. Next, we describe the mechanism for selecting the advertiser and his sponsored content, as well as establishing how payments are determined. Subsequently, we will define the utilities of various stakeholders, specifically the user, the platform, and the advertisers. Together, these elements will constitute the framework of a game, where each advertiser's strategy of every player is his bid on a specific sponsored content to be displayed to the user.

% The advertisers bid on a specific sponsored content generated as defined above. 
}
\begin{definition}
  \label{def:setup}
A question-specific setup of a question-answering (QA) platform with sponsored answers is a tuple\\
$\setup$. $\question$ and $\pure$ are a question and the organic
answer generated for it, respectively. $N:=\{1,2,...,n\}$ is a set of advertisers where
$n\geq2$. $\{\adi\}_{\inN}$ are the advertisers' ads for $\question$.
$\SCi$ is the sponsored answer generated by fusing the organic answer $\pure$ with ad $\adi$. $\valuei\in\mathbb{R_{+}}$ is the value for advertiser $i$
of displaying the sponsored answer $\SCi$ to the user.
\end{definition}

We assume the QA platform has to select a single sponsored answer from
$\{\SCi\}_{\inN}$ to show the user. This practice is aligned with a setting where the platform shows a single organic answer.
%Given the possible sponsored contents $\{\SCi\}_{\inN}$, the platform needs to decide which sponsored content will be showed to the user.
As in classical sponsored search \cite{Jansern+Nullen:08a}, the
platform is incentivized to receive a payment based on the content it
presents to users. We now turn to define the utilities of the
stakeholders --- users, advertisers and the platform itself --- and
propose an auction mechanism to determine the winner and her
payment.

%In general, utility definition calls for a measure of the extent to
%which a sponsored answer is "good" for someone (user or
%advertiser). To this end, we can assume some inter-text similarity
%measure, and use its values as proxies for ``goodness''.

%For example,
%advertiser $i$ might opt to have the sponsored answer $\SCi$ as
%similar as possible to the ad $\ado$ from which it was generated by
%fusion with the organic answer ($\pure$).
%We hasten to point out, however, that the results we present up to Section \ref{sec:unigram} do not depend on a specific definition of a utility function. We use textual similarity as a running example for how utilities in our setting might be perceived and defined.

%We utilize the
%similarity metric between texts defined earlier, now incorporating it
%as a component to gauge an individual's utility from a particular
%sponsored content. For instance, the advertiser $i$ wish the sponsored
%content contains as much information as defined in his original
%advertisement $d_i$.

%The stakeholders in our game (auction) are the (i) advertisers who opt
%for a minimal payment for the suggested sponsored answer that will
%result in winning the auction; ii) the user who wants to receive the
%most relevant answer to her question; and, (iii) the QA platform which
%strives to maximize both user and advertisers satisfaction.

We focus on a setup where the (game) strategy $s_i$ of every
advertiser $i$ is a bid $\bidi$ ($\in \mathbb{R_{+}}$) for showing
the sponsored answer $\SCi$ to the user. The set of {\em strategy
  profiles} of the advertisers is $S=S_1 \times S_2 \times ... \times
S_n$ where $S_i$ is the set of all possible bids of advertiser $i$;
namely, $\mathbb{R_{+}}$. We adapt GSP (generalized second price auction) \citep{varian_position_2007} to our setting for the allocation and payment procedures.

The platform's value for a sponsored answer $\SCi$ is an aggregation of the user utility and the advertiser's bid on $\SCi$:
%Both definitions rely on a
%function $S(x,y)$ of the satisfaction from text $x$ given
%that the goal is text $y$. An inter-textual similarity measure, as noted above, can serve to construct such a satisfaction function.
\begin{definition}
  \label{def:platformValue}
The prospect \textbf{Platform Value} (henceforth, platform value in short) from displaying the sponsored answer $\SCi$ is defined as $\VP(\SCi, \question, \bidi) = \userUtil{\SCi}{\question} + \bidi$, where $\userUtil{\SCi}{\question}$ is the user utility attained by seeing the sponsored answer $\SCi$ in response to her question ($\question$), and $\bidi$ is the bid of the advertiser $i$ on the sponsored answer $\SCi$.
\end{definition}
We use the term ``prospect'' to emphasize that the payment received by
the platform can be lower than $\bidi$ as explained below. We next
define social welfare with respect to a sponsored answer as the sum of the user
utility and the value for the advertiser whose sponsored answer was
selected\footnote{It is standard practice in economics to account for the utilities of different stakeholders.}. Using alternative definitions of social welfare and analyzing the resultant auction is left for future work. In standard GSPs for sponsored search
\cite{varian_position_2007}, the social welfare is a weighted sum of
advertisers' values (cf., Equation \ref{eq:social}) where the weights
represent click probabilities\footnote{In standard GSPs for sponsored
  search, ads are ranked. In our setting, a single ad is presented,
  and only its advertiser has potential gain.}. Here we do not account
for click probabilities. We do take into account, in contrast to
standard GSPs, the user's utility.
\begin{definition}\label{sw_definition}
  The \textbf{Social Welfare} from displaying the sponsored answer $\SCi$ is $SW(\SCi, \question, \bidi) = \userUtil{\SCi}{\question} + \valuei$.
  %where $S(\SCi, \pure)$ is the user satisfaction with the sponsored answer $\SCi$ and $S(\SCi, \adi)$ is the advertiser's satisfaction with the sponsored answer $\SCi$.
\end{definition}
Note that the social welfare is composed of a user part (utility) and advertiser part (value). Indeed, the platform aims at satisfying both users and advertisers.

\omt{
In those definitions we assume the values of both the advertiser $i$ and the user from showing the sponsored content $\SCi$ are proportional to the similarities to the original ad $\adi$ and the pure answer $\pure$, respectively. Note, that our framework and the coming proof holds for a more general case which the advertiser value and user value are not necessarily defined as the similarity between the sponsored content to the pure answer and the original advertisement, respectively. Specifically, $\simdi$ and $\simdo$ can be replaced with general $\valuei$, $u_i(\SCi)\in\mathbb{R_{+}}$ and all the formalism holds.

***** Moshe: not sure there are such "winning references"; in fact about distance  to information need it is more like in dense retrieval, isn't it? We can also cite the paper in arxiv on the search of stability ****
*** Tommy: I will add search for stability**

**** Moshe: the below is a bit too apologetic, We use strategies as bids and the GSP format  as the allocation and payment procedures following Varian's original work (cite the paper on position auctions). We should though say that the allocation and payment in our mechanism  would depend on the lambdas, as any sets of Lambdas define a different game. Hope this will clear... ***** 
***Tommy: i rephrase everything according to that***
}
The advertiser selected as the {\em winner} of the auction, $\winner$, is the one whose sponsored answer and bid maximize the platform value:
\begin{equation}
  \label{eq:winner}
    \winner :=  arg\max_i \VP(\SCi, \question, \bidi).
\end{equation}
The platform shows the sponsored answer $\SCis$.

In GSP auctions \cite{varian_position_2007}, the winner pays the
minimal payment $\payis\leq \bidis$ needed to surpass the second best
platform value, attained by advertiser $\second$. Formally,\\
$\VP(\SCit, \question, \bidit) \leq \VP(\SCis, \question, \bidis)$ and\\ $\VP(\SCit, \question, \bidit) = \VP(\SCis, \question, \payis)$.
Based on these definitions, we arrive to the following result:
\begin{lemma}\label{payment}
The payment $\payis$ of the winner $\winner$ is$~$ $\bidit +  \userUtil{\SCit}{\question}  -  \userUtil{\SCis}{\question}$.
\end{lemma}
\begin{proof}
The winner $\winner$ pays the minimal $\payis$ s.t. the resultant platform value is the second best. 
Since\\ $\VP(\SCit, \question, \bidit) = \VP(\SCis, \question, \payis)$, then $\userUtil{\SCit}{\question} + \bidit  =  \userUtil{\SCis}{\question} + \payis$
and therefore $\payis = \bidit + \userUtil{\SCit}{\question}  -  \userUtil{\SCis}{\question}$.
\end{proof}
We now define the utilities of advertisers. Note that the utility
of an advertiser depends on the bids of all other advertisers, not
only her own; specifically, the payment of the winner depends on the bid of another advertiser. Indeed, recall from Section \ref{sec:gt} that the
utility of a player in a game depends on the strategies employed by
all players.

Following common practice in auction theory \cite{Myerson:81a}, we
assume quasi-linear utility functions of the advertisers: we
substract the payment from the value to determine the advertiser
utility if she is the winner, and set the utility to $0$ otherwise.
\begin{definition}
The utility of advertiser $i$ is:
%$U_i(B,\{\valuei\}_{\inN}, \{\adi\}_{\inN}, \{\SCi\}_{\inN}, \pure, \voc) = $ \\
\begin{equation}
\adUtil{i}(\SCi,\valuei,\{\SCj\}_{j \ne i},\{\valuej\}_{j \ne i}) = 
\begin{cases}
    \valuei - \payi & \text{if i wins},\\
    $0$  &\text{otherwise};
\end{cases}
\end{equation}
$\payi$ is  $i$'s payment in case she wins.
\end{definition}

%\begin{figure*}[!htb]
%   \begin{minipage}{0.9\textwidth}
%     \centering
%     \includegraphics[width=1\linewidth]{figs/sponsored_content_framework.jpg}
%     \caption{Our Framework of sponsored content in QA systems. Given a question from the user (1) and a pure answer (2) there is $n$ advertisers that wish to display their advertisement to the user (3). The platform generate $n$ possible sponsored content, one for each advertiser (4) and propose them to the advertisers. The advertisers calculate their value from displaying their sponsored content (5) and bid on this (6). According to the auction mechanism, the winning advertiser display its sponsored content (7) and pay the platform (8). Even though we use a specific modeling of the values of the advertiser and user: $\valuei=\simdi$ } $u_i(\SCi)=\simdo$, respectively, The whole framework, and proofs hold for a more general modeling of those components.\label{framework}
%   \end{minipage}
%\end{figure*}

The auction setting we defined can be analyzed as a game where the
advertisers' strategies are their bids and they strive to maximize
their utility. We next show that advertisers have an incentive to
behave truthfully, that is, bid the true value of their sponsored
answer. More specifically, this bidding strategy is a dominant
strategy. As a result, the game (auction) has a Nash equilibrium
(i.e., the game is stable). This stability is highly important. A case in point, in auctions which are not incentive-compatible (i.e, bidders are not incentived to bid their true value) {\em shading} phenomena \cite{Nautz+al:97a} are prevalent: bidders continuously addapt their bids (specifically, lower than the true value) to maximize gain.
\begin{lemma}
  \label{dominant}
$\bidi=\valuei$ is a dominant strategy of advertiser $i$.
\end{lemma}
\begin{proof}
%Let fix $i$, and remind its value $v_i$. We fix $b_{-i}$ the bids of all other players. 
We show that the utility of player $i$ is maximized when $\bidi=\valuei$.
Let $M=max_{j\neq i} \VP(\SCj, \question, \bidj)$ be the maximal platform value attained by selecting a player $i' \neq i$.
There are two options for the utility of advertiser $i$:
\begin{itemize}
    \item $i$ is not the winner: $\VP(\bidi) \leq M$. In this case, the utility of player $i$ is $U_i(\bidi) = 0$ and $\bidi=\valuei$ is a dominant strategy.
    \item $i$ is the winner: $\VP(\bidi) > M$. Hence, $j=\second$ and $M = \bidit + \userUtil{\SCit}{\question}$. By Lemma \ref{payment},  $\payi = M - \userUtil{\SCi}{\question}$. The utility of $i$ is then $\valuei - \payi = \valuei - M + \userUtil{\SCi}{\question} $.
We analyze two cases :
\begin{itemize}
    \item if $\valuei - M + \userUtil{\SCi}{\question}$ < 0 then $\valuei
      < \userUtil{\SCit}{\question} + \bidit -
      \userUtil{\SCi}{\question}$ and $\valuei +
      \userUtil{\SCi}{\question} < \bidit +
      \userUtil{\SCit}{\question}$. If $\bidi \leq \valuei$, we get
      $\bidi + \userUtil{\SCi}{\question} < \bidit +
      \userUtil{\SCit}{\question}$. Consequently, $\second$ is the
      winner and we get a contradiction.
      %$i$ has a $0$ utility attained by 
      %setting $\bidi=\valuei$.
      \item if $\valuei - M + \userUtil{\SCi}{\question} \geq 0$, then $i$ can maximize her utility by winning the game having $\valuei - \payi \geq 0$; to this end, $i$ can bid $\bidi = \valuei$ since $\valuei + \userUtil{\SCi}{\question} \geq M$. 
\end{itemize}
\end{itemize}
\end{proof}
It is easy to show that the utility of the winner is the difference between the social welfare if her sponsored answer is selected and the social welfare if the sponsored answer of the second best advertiser (in terms of platform value) is selected. This is a property of VCG auctions \cite{MasColell+al:95a}. Formally,
\begin{proposition}\label{obs_sw}
  Assume all advertisers play their the dominant strategy: $\bidi = \valuei$. The utility of the winner $\winner$ is:
  %the difference between the social welfare (as defined in dDfinition \ref{sw_definition})  given the sponsored result of the winner and the social welfare of the second best advertiser (w.r.t. the platform value):
  $\valueis - \valueit +  \userUtil{\SCis}{\question} - \userUtil{\SCit}{\question} = SW(\SCis, \question, \bidis) - SW(\SCit, \question, \bidit)$.
\end{proposition}

\section{From Ads to Sponsored Answers}
\label{sec:unigram}
Heretofore, we treated the fusion of the organic answer
$\pure$ and the ad $\adi$ to yield a sponsored result, $\SCi$, as a
black box. That is, advertiser $i$ bids on $\SCi$ based on the value
$\valuei$ it attributes to it. Then, the QA platform selects the
sponsored answer that yields the highest platform value (see
Definition \ref{def:platformValue}). We now turn to describe possible fusion methods and discuss their implications.

\subsection{Sponsored Answer Generation Using Large Language Models}
\label{sec:fuseLLM}
A possible approach to fusing the organic answer $\pure$ and an ad
$\adi$ is to request (via a prompt) a large language model (LLM; e.g.,
GPT) to perform the fusion \cite{zelch_user_2024}. However, it is impossible to
make theoretical statements about the resultant sponsored answer,
$\SCi$. 

Furthermore, it is highly difficult, to impossible, to quantitatively
control the relative emphasis on $\pure$ versus $\adi$ in a
prompting-based fusion approach as just described. This aspect can be
quite important for both the advertiser and the user of the QA
platform: the advertiser strives to have the sponsored answer $\SCi$ similar to $\adi$ and the user probably wants it to be similar to the
organic answer, $\pure$. We re-visit this point below.

We now consider an alternative approach to fusion inspired by a recent
proposal of integrating multiple LLMs for content generation \cite{duetting_mechanism_2023}. Suppose advertiser $i$ generates the (textual) ad $\adi$
using an LLM, denoted $\lm^{ad}$. Suppose also that the QA platform
uses an LLM, $\lm^{org}$, to generate the organic answer $\pure$ for
the question $\question$ which is used as the prompt. Assume that both
LLMs are autoregressive. That is, given a sequence of tokens $x$
(composed of the prompt and tokens already generated), there is a
probability distribution over the token vocabulary $\voc$ from which
the next token is sampled. Formally, $p(t\vert \lm,x)$ is the
probability of token $t$ given the LLM $\lm$ and the sequence $x$
already generated. Then, we can fuse the organic answer $\pure$ and
the ad $\adi$ by (i) defining a next-token generation distribution:
$\lambda p(t\vert \lm^{org},x) + (1-\lambda) p (t\vert \lm^{ad},x)$,
where $\lambda$ is a parameter controlling the relative emphasis on
the organic answer and the ad; and (ii) using the original prompts in
each of the two LLMs: the request to generate the ad ($\lm^{ad}$) and the question ($\lm^{org}$).

There are three main challenges embodied in the fusion process just
described. First, there is no guarantee about the quality of
the generated sponsored answer, $\SCi$, since we mix distributions at the token level. Second, we assumed that an ad was generated using an LLM which is not necessarily the case. Third, we still cannot theoretically reason about the generated sponsored answer and formally use it in the advertiser value function and user utility function defined above.

To facilitate a first step towards a formal end-to-end (theoretical)
analysis of the process of fusing $\pure$ and $\adi$ to produce
$\SCi$, and using $\SCi$ in the auction defined in Section
\ref{sec:auctions}, we subscribe to the language modeling framework to
retrieval \cite{Ponte+Croft:98a,Lafferty+Zhai:01a}. Specifically, we
use unigram language models based on a multinomial distribution
\cite{Ponte+Croft:98a,Song+Croft:99a,Lafferty+Zhai:01a}.

\subsection{Unigram Language Models}
Let $x$ be a text and $t$ a term (token) in the vocabulary
$\voc$. $p_x(t) := \frac {\tf{t}{x}}{\sum_{t' \in x} \tf{t'}{x}}$ is
the maximum likelihood estimate of $t$ with respect to $x$ assuming a
multinomial distribution; $\tf{t}{x}$ is the number of times $t$
appears in $x$. $p_x(\cdot)$ is the unsmoothed unigram
language model induced from $x$. The probability assigned to a term
sequence $t_1,\ldots,t_n$ is $p_x(t_1, \ldots, t_n) := \prod_{i=1}^{n} p_x(t_i)$. Often, unigram language models 
are smoothed to avoid the zero probability problem \cite{Zhai+Lafferty:01a}. For our formal analysis herein, smoothing is not required. We hasten to point out however that
our findings also hold for smoothed language models.

Given an organic answer $\pure$ and an ad $\adi$, we can fuse the
unigram language models induced from them using a linear mixture to
produce a unigram language model from which a sponsored answer will be
generated: $p_{i}^{spon}(t) := \lambda_i p_{\pure}(t) + (1-\lambda_i)
p_{\adi} (t)$. Note that $\lambda_i$ is advertiser-specific and can
potentially be negotiated between the advertiser and the QA platform
before the auction takes place. We can now generate a specific
sponsored answer $\SCi$ of length $k$ by sampling $k$ times from
$p_i^{spon}(\cdot)$. The {\em mean} number of occurrences of token $t$
in a document of length $n$ generated using $p_{i}^{spon}(\cdot)$, i.e.,
in a sponsored answer, is $np_i^{spon}(t)$.\footnote{This is the mean in
  a multinomial distribution.} Herein, we use $\SCi$ to denote the
``mean'' document which is composed of these mean number of
occurrences of tokens, henceforth simply referred to as the sponsored
answer. As was the case of fusing LLMs at the next-token generation
level, there is no guarantee about the quality of the generated
sponsored answer. We hasten to point out that our use of unigram
language models is intended to facilitate the formal analysis of
the auction mechanism given a specific fusion method.

\myparagraph{Advertiser's value, user's utility and text similarity}
We recall that advertiser $i$ is interested in having her ad $\adi$
presented for the question $\question$. Instead, the platform suggests
showing the sponsored answer $\SCi$ which is the result of fusing
$\adi$ with the organic answer $\pure$. Naturally, the advertiser
strives to have the sponsored answer as similar as possible to her
ad. Hence, we define the advertiser's value function used in
Definition \ref{def:platformValue} as
\begin{equation}
  \label{eq:advertiserVal}
  \valuei := \simFn{\SCi}{\adi},
  \end{equation}
where $\simFnSymbol$ is an inter-text similarity measure. We note that
the use of a similarity function here serves two roles. The first is
having common grounds for the advertiser's value function and user's
utility function: both are defined in terms of textual
similarity. (See the below.) Second, the specific similarity function
we use below allows for convenient mathematical treatment. We could have used a variety of alternatives for defining the value function; e.g., based on textual entailment or more evolved natural language inference approaches.

Similarly, we assume that the goal of the user who posted the question
$\question$ so as to (presumably) receive an organic answer is to have the sponsored answer $\SCi$ as similar as possible to the organic answer $\pure$. Accordingly, we define the user utility function used in Definition \ref{def:platformValue} as
\begin{equation}
  \label{eq:userUtility}
  \userUtil{\SCi}{\question} := \simFn{\SCi}{\pure}.
\end{equation}

We define the similarity of texts $x$ and $y$ based on the cross
entropy measure (cf. \cite{Lafferty+Zhai:01a}):
\begin{equation}
  \label{eq:sim}
  \simFn{x}{y} := 2A - CE (p_x(\cdot) || p_y(\cdot)) - CE (p_y(\cdot)|| p_x(\cdot));
\end{equation}
$CE (p_x(\cdot) || p_y(\cdot)) = \sum_{t \in \voc} p_x(t) log p_y(t)$;\footnote{We use cross entropy rather than KL divergence because the cross entropy is linear in its left argument. Note that the resultant similarity function is concave.} lower cross entropy corresponds to increased
similarity\footnote{As noted above, usually language models are smoothed so as
to avoid zero probabilities (e.g., in cross entropy computation). For our purposes (constructive proofs presented below), smoothing is not needed as there are no cases of zero probabilities.};
$A \in \mathbb{R_{+}}$ is used to ensure that the
similarity value is in $\mathbb{R_{+}}$: $A := \max_{j_1,j_2} CE(p_{d_{j_1}}(\cdot)||p_{d_{j_2}}(\cdot)) + CE(p_{d_{j_2}}(\cdot)||p_{d_{j_1}}(\cdot))$ where
$d_{j_1}, d_{j_2} \in \{{\pure}\} \cup \{d_i\}_{i\in N}$.

\myparagraph{Auction Analysis with Unigram Language Models} We now
turn to show that in the unigram-language-model setting described
above, the advertiser $i$ who wins the auction is not necessarily the one with
the maximal value $\valuei$. Since by Lemma \ref{dominant} the
dominant strategy is to set the bid as the value, we arrive to an
interesting result: the advertiser who placed the highest bid is not
necessarily the one who wins the auction. This result is in contrast to the state-of-affairs in many other auction mechanisms, and is due to the QA platform's goal to satisfy both users and advertisers.

%including the basic Generalized Second Price (GSP) auction \cite{varian_position_2007}.\footnote{In more evolved GSPs for sponsored search, the bid is combined with a quality measure so as to determine the winner \cite{varian_position_2007}.}

\begin{proposition}\label{surplus_obs1}
The winner of the auction is not necessarily advertiser $i$ whose
value $\valuei$ is the maximum with respect to all advertisers'
values. Accordingly, the winner is not necessarily the advertiser who placed the highest bid.
\end{proposition}
\begin{proof}
  Consider a 2-advertisers setting with a two-terms vocabulary: $\voc=\{a,b\}$.
  Suppose $\pure$, $d_1$ and $d_2$ are all sequences of $k$ terms.
  Suppose (1-$\epsilon$) of the tokens in $\pure$ are $a$ and $\epsilon$ are $b$ ($\epsilon \in [0,1]$). Consequently: $p_{\pure}(a) = 1 - \epsilon$ and $p_{\pure}(b) = \epsilon$. For $d_1$ we assume: $p_{d_1}(a) = \epsilon$ and $p_{d_1}(b) = 1- \epsilon$. For $d_2$ we assume: $p_{d_2}(a) = p_{d_2}(b) = 0.5$.

  Following the fusion approach described above, $p^{s}_{d_i}(t) = \lambda_i p_{d_0}(t) + (1-\lambda_i) p_{d_i}(t)$ for $i \in \{1,2\}$ and $t \in \{a,b\}$. We set $\lambda_1=\epsilon$ and $\lambda_2 = 1-\epsilon$. It can be shown that for $\epsilon \rightarrow 0$, $\simFn{d_1^{s}}{d_1} =  v_1(d_1^{s}) > \simFn{d_2^{s}}{d_2} = v_2(d_2^{s})$ and $v_2(d_2^{s}) + \simFn{d_2^{s}}{d_0} > v_1(d_1^{s}) + \simFn{d_1^{s}}{d_0}$; i.e., $v_2(d_2^{s}) + \userUtil{d_2^{s}}{\question}  > v_1(d_1^{s}) + \userUtil{d_1^{s}}{\question}$. Thus, $d_2$ wins the auction but the value for $d_1$ is higher. The full proof is provided in Appendix \ref{surplus_obs1_appendix}. Since by Lemma \ref{dominant} bidding the value is a dominant strategy, we get that the winner of the auction is not the one who placed the highest bid.
\end{proof}
We next show that the user utility is not necessarily the maximal
possible. This is, again, due to the fact that the QA platform aims to satisfy
both users and advertisers.
\begin{proposition}\label{surplus_obs2}
  If advertiser $i$ won the auction,
  %it is not necessarily the case
  %that
  the resultant user utility, $\userUtil{d_i^{s}}{\question}$,
  is not necessarily the maximal with respect to that attained by selecting other advertisers (sponsored answers).
\end{proposition}
\begin{proof}
  We consider the same 2-advertisers setting as in Proposition \ref{surplus_obs1} with the same ads and organic answer. The difference is that now we set $\lambda_1= 1 - \epsilon$ and $\lambda_2=0.5$.

  %
%  with a two-terms vocabulary:
%  $\voc=\{a,b\}$. Suppose $\pure$, $d_1$ and $d_2$ are all sequences of $k$ terms. Suppose (1-$\epsilon$) of the tokens in $\pure$ are $a$ and $\epsilon$ are $b$ ($\epsilon \in [0,1]$). Consequently: $p_{\pure}(a) = 1 - \epsilon$ and $p_{\pure}(b) = \epsilon$. For $d_1$ we assume: $p_{d_1}(a) = \epsilon$ and $p_{d_1}(b) = 1- \epsilon$. For $d_2$ we assume: $p_{d_2}(a) = p_{d_2}(b) = 0.5$.
  % Following the fusion approach described above, $p_{d_i}(t) = \lambda_i p_{d_0}(t) + (1-\lambda_i) p_{d_i}(t)$ for $i \in \{1,2\}$ and $t \in \{a,b\}$.
  % Let $\lambda_1= 1 - \epsilon$ and $\lambda_2=0.5$.
  It can be shown that for $\epsilon \rightarrow 0$, $\userUtil{d_1^{s}}{\question} = \simFn{d_1^{s}}{\pure} >  \simFn{d_2^{s}}{\pure} = \userUtil{d_2^{s}}{\question}$ and  $\valueb + \userUtil{d_2^{s}}{\question}  > \valuea + \userUtil{d_1^{s}}{\question}$. That is, $i=2$ is the winner of the auction but the resultant user utility is lower than that for advertiser 1. The full proof is in Appendix \ref{surplus_obs2_appendix}.
\end{proof}

%% file: preliminaries.tex
%\section{Background}
\subsection{Game Theory}
\label{sec:gt}
We now briefly review some basic concepts in game theory. 
\begin{definition}
  An n-players game is a tuple $G=(\{S_i\}_{i\in N}, \{U_i\}_{i\in N})$. $S_i$ is the set of strategies of player $i$.
  %\subseteq
  %\mathbb{R_{+}}$
  %  is player $i$'s set of possible strategies\footnote{$\mathbb{R_{+}}$ is the set of non-negative real numbers.}.
  $S=S_1 \times S_2 \times ... \times S_n$ is the set of \textbf{strategy profiles} in the game. $U_i:S\rightarrow \mathbb{R_{+}}$ is the \textbf{utility function} of player $i$.\footnote{$\mathbb{R_{+}}$ is the set of non-negative real numbers.}
\end{definition}

Each player in a game aims to maximize her utility. Note that the utility
depends not only on her strategy but also on the strategies of other
players. Consider the following 2-players game:
$S_1=S_2=[0,1]$, $U_1(s_1, s_2) = 1 + s_1^2 + s_2^2, U_2(s_1, s_2) = 2 + s_1^2 + s_2^2$. If the players select $s_1=1$ and $s_2=1$, then their utilities are
$U_1(1,1) = 3$ and  $U_2(1,1) = 4$.

A fundamental characterization of games is their stability or lack
thereof. A {\em stable} strategy profile of a game is a profile where
no player has an incentive to deviate from her strategy. A well known example of a stable profile is Nash equilibrium. Let $S_{-i} = S_1\times...\times S_{i-1} \times S_{i+1} \times ... \times S_n$ denote the strategy profile of all the players except $i$. Then,
%This principle is commonly known as the
%game's solution.
\begin{definition}
  \label{def:nash}
 A strategy profile $s=(s_i, s_{-i})$ where $s_i \in S_{i}, s_{-i} \in S_{-i}$ is a Nash equilibrium  if$~~$ $\forall i \in N$ and $\forall s_i' \in S_i$, $U_i(s_i', s_{-i}) \leq U_i(s_i, s_{-i})$.\footnote{We focus on pure strategies and pure Nash equilibrium. Mixed strategies which are distributions over pure strategies are outside the scope of this paper.}
\end{definition}
Considering the game example from above, it is easy to show that the strategy profile $s_1=1, s_2=1$ is a Nash equilibrium.

%Regarding Our example, The game has a pure Nash equilibrium: the strategy profile $s_1=1, s_2=1$. As a proof, let observe a deviation of player $i=1$ . In the defined game, $0 \leq s_1 \leq 1, 0 \leq s_2 \leq 1$, and assuming player $i=1$ deviates to strategy $s<1$, then he reduces his utility:  $U_1(s, 1) = 1 + s^2 + 1^2 < U_1(1, 1) = 1 + 1^2 + 1^2$. The same proof holds for the second player. Note that in general, a pure Nash equilibrium in the games we have defined, does not necessarily occur.

A {\em dominant strategy} of a player is a strategy that is better than any other strategy regardless of the strategies of other players:
\begin{definition}
$s_i \in S_i$ is a dominant strategy of player $i$ if$~~$ $\forall s_{-i} \in S_{-i}, \forall s'_{i} \in S_i, U_i(s'_{i}, s_{-i}) \leq U_i(s_i, s_{-i})$.
\end{definition}
Note that players' strategies in a Nash equilibrium (Definition \ref{def:nash}) need not
necessarily be dominant strategies. In our example game, $s_1=1$ is a
dominant strategy of the first player, since $U_1(1,s_2) \geq U_1(s_1,
s_2)$~~$ \forall s_1,s_2 \in [0,1]$.

A common measure of the "goodness" of a strategy profile $s$ is {\em
  social welfare}, often defined as the sum of the players' utilities:
\begin{equation}
  \label{eq:social}
   SW(s) = \sum_{i \in  N} U_i(s).
\end{equation}

%% file: conc.tex
\section{Conclusions and Future Work}
\label{sec:conc}
We presented a novel formal study of a platform for sponsored question answering
(QA). The platform is based on fusing an organic answer to a question
with an ad so as to produce a sponsored answer. Advertisers bid on
their corresponding sponsored answers. Inspired by principles of Generalized Second
Price Auctions (GSPs) \cite{,edelman_internet_2007,varian_position_2007}, the platform selects the sponsored answer to
show the user and sets the payment for the selected advertiser.

We prove that a dominant strategy of advertisers is to bid on their
true value of the sponsored answer which is a property of VCG \cite{MasColell+al:95a} auctions. The result is that the QA setting reaches a stable state
(equilibrium) where advertisers have no incentive to continuously
change their bids (a.k.a., shading). We also formalize the notion of
social welfare and show that the utility of the advertiser who wins the auction is the difference between the social welfare attained when presenting her sponsored answer and the social welfare of presenting the sponsored answer of the second best advertiser whose bid is the payment the winner has to pay.

Our general analysis of the auction is not committed to a specific
approach to fusing an organic answer and an ad. To theoretically
analyze end-to-end the process of fusing an organic answer with an ad
and apply the auction, we use unigram language models. We prove that
the auction winner is not necessarily that with the highest bid. We
also prove that the attained user utility is not the maximal possible
with respect to selecting other advertisers. This result is due to the fact that we design the QA platform to satisfy both users and advertisers.

In the unigram-language-model analysis, we used an inter-text similarity measure as a basis for defining advertisers' value functions and users' utility functions. For future work we plan to study the effect of using alternative value and utility functions.

Obviously, there are far reaching ethical considerations involved in a
setting where organic content is mixed with ads; cf., the
state-of-affairs in native advertising where users often fail to
distinguish ads from organic content even when these are not fused
\cite{Amazeen+Wojdynski:20a}. We refer the reader to discussion of ethical issues in recent work on fusing organic answers and ads \cite{zelch_user_2024}. We take the stand that the resultant sponsored answers should be explicitly marked as such giving users the alternative to consume only organic answers.

%The task of fusing an organic question with an ad deserves an in-depth
%future study. Such a study calls for the creation of evaluation datasets.

%% file: ack.tex
\smallskip
\myparagraph{Acknowledgments} We thank the reviewers for their comments.
The work by Moshe Tennenholtz was supported by funding from the
European Research Council (ERC) under the European Union's Horizon
2020 research and innovation programme (grant agreement 740435).

%% file: appendix.tex
\appendix

\section{Proofs}
In what follows $CE(d_a||d_b)$
stands for $CE(p_{d_a}(\cdot)||p_{d_b}(\cdot))$. In both proofs below, we omit in the computations factors which are $O(\epsilon)$ or $O(\epsilon \log \epsilon)$ since we assume $\epsilon \rightarrow 0$.

\subsection{Proof of Proposition \ref{surplus_obs1}} \label{surplus_obs1_appendix}
 Recall that%$\pure, \ada, \adb$:\\
$\un{\pure}{a} = 1-\epsilon$, $\un{\pure}{b}
= \epsilon$, $\un{\ada}{a} = \epsilon$, $\un{\ada}{b} = 1-\epsilon$,
$\un{\adb}{a} = 0.5$, $\un{\adb}{b} = 0.5$. Since $\lambda_1
= \epsilon, \lambda_2 = 1-\epsilon$:

$\un{\SCa}{a} = \lambda_1 \un{\pure}{a} + (1-\lambda_1) \un{\ada}{a} = \epsilon
(1-\epsilon) + (1-\epsilon) \epsilon = 2\epsilon(1-\epsilon)$.

$\un{\SCa}{b} = \lambda_1 \un{\pure}{b} + (1-\lambda_1) \un{\ada}{b}
= \epsilon \cdot \epsilon + (1-\epsilon) (1-\epsilon) = \epsilon^2 +
(1-\epsilon)^2$.

$\un{\SCb}{a} = \lambda_2 \un{\pure}{a} +
(1-\lambda_2) \un{\adb}{a} = (1-\epsilon) (1-\epsilon)
+ \epsilon \cdot 0.5 = (1-\epsilon)^2 + \frac{\epsilon}{2}$.

$\un{\SCb}{b} = \lambda_2 \un{\pure}{b} + (1-\lambda_2) \un{\adb}{b} =
(1-\epsilon) \cdot \epsilon+ \epsilon \cdot 0.5 = \epsilon \cdot
(1-\epsilon) + \frac{\epsilon}{2}$.

We now compute $\valuea, \valueb, \userUtil{\SCa}{\question}, \userUtil{\SCb}{\question}$ using the fact that $CE$ is linear in its left argument:

$\valuea = \simFn{\SCa}{\ada} = 2A - CE(\SCa||\ada) -CE(\ada||\SCa) = 2A -\epsilon CE(\pure||\ada) - (1-\epsilon) CE(\ada||\ada) + \epsilon log(2\epsilon(1-\epsilon)) + (1-\epsilon)log((1-\epsilon)^2+\epsilon^2) = 2A - CE(\ada||\ada)$.

$\valueb = \simFn{\SCb}{\adb} = 2A -CE(\SCb||\adb) -CE(\adb||\SCb) = 2A -(1-\epsilon) CE(\pure||\adb) -\epsilon CE(\adb||\adb) + 0.5 log((1-\epsilon)^2+\frac{\epsilon}{2}) + 0.5 log(\epsilon(1-\epsilon)+\frac{\epsilon}{2}) = 2A - CE(\pure||\adb) + 0.5 log(\frac{3\epsilon}{2} -\epsilon^2)$.

$\userUtil{\SCa}{\question} = \simFn{\SCa}{\pure} = 2A - CE(\SCa||\pure) -CE(\pure||\SCa) = 2A - \epsilon CE(\pure||\pure) - (1-\epsilon) CE(\ada||\pure) + (1-\epsilon) log(2\epsilon(1-\epsilon)) + \epsilon log((1-\epsilon)^2 + \epsilon^2) = 2A -CE(\ada||\pure) + log(2\epsilon(1-\epsilon))$.

$\userUtil{\SCb}{\question} = \simFn{\SCb}{\pure} = 2A - CE(\SCb||\pure) - CE(\pure||\SCb) = 2A - (1-\epsilon) CE(\pure||\pure) - \epsilon CE(\adb||\pure) + (1-\epsilon) log((1-\epsilon)^2 + \frac{\epsilon}{2}) + \epsilon log(\epsilon(1-\epsilon)+\frac{\epsilon}{2}) = 2A - CE(\pure, \pure)$.

We now compute $CE(\ada||\ada), CE(\ada||\pure), CE(\pure||\adb), CE(\pure, \pure)$:

$CE(\ada||\ada) =-\epsilon log(\epsilon) -(1-\epsilon)log(1-\epsilon)=0$.

$CE(\ada||\pure) =-\epsilon log(1-\epsilon) -(1-\epsilon)log(\epsilon) = -(1-\epsilon)log(\epsilon) = -log(\epsilon)$.

$CE(\pure||\adb) = -(1-\epsilon)log(0.5) -\epsilon log(0.5) = -log(0.5)$.

$CE(\pure||\pure) = -(1-\epsilon)log(1-\epsilon) -\epsilon log(\epsilon) = 0$.

We show that the value of the first advertiser is larger than that of the second advertiser:

$\valuea > \valueb \Leftrightarrow 2A - CE(\ada||\ada) >  2A - CE(\pure||\adb) + 0.5 log(\frac{3\epsilon}{2} -\epsilon^2) \Leftrightarrow 2A > 2A +log(0.5) + 0.5 log(\frac{3\epsilon}{2} - \epsilon^2) \Leftrightarrow 0 > log(0.5) + 0.5log(\frac{\epsilon}{2}) + 0.5log(3-2\epsilon) \Leftrightarrow -log(0.5) - 0.5log(3) > 0.5log(\frac{\epsilon}{2}) \Leftrightarrow \frac{8}{3} > \epsilon$.

We now turn to determine the winner. The prospect platform value for each advertiser is:

$\VP(\SCa, \question, \valuea)=\valuea + \userUtil{\SCa}{\question} = 4A + log(\epsilon) + log(2\epsilon(1-\epsilon))= 4A + log(\epsilon) + log(2) + log(\epsilon) + log(1-\epsilon)= 4A + 2log(\epsilon) + log(2)$.

$\VP(\SCb, \question, \valueb)=\valueb + \userUtil{\SCb}{\question} = 4A + log(0.5) + 0.5 log(\frac{3\epsilon}{2} - \epsilon^2) = 4A + log(0.5) + 0.5 log(\epsilon) + 0.5 log(1.5 - \epsilon)=4A + log(0.5) + 0.5 log(\epsilon) + 0.5  log(1.5)$.

Thus, the winner is the second advertiser since:

$\valuea + \userUtil{\SCa}{\question} < \valueb + \userUtil{\SCb}{\question} \Leftrightarrow 4A + 2log(\epsilon) + log(2) < 4A + log(0.5) + 0.5 log(\epsilon) + 0.5  log(1.5) \Leftrightarrow 1.5 log(\epsilon) < -log(2) + log(0.5) + 0.5 log(1.5) \Leftrightarrow \epsilon < \frac{\sqrt[3]{6}}{4}$.

\subsection{Proof of Proposition \ref{surplus_obs2}} \label{surplus_obs2_appendix}
Recall that we use the same documents as in
Proposition \ref{surplus_obs1}: $\un{\pure}{a} = 1-\epsilon$,
$\un{\pure}{b}=\epsilon$, $\un{\ada}{a} = \epsilon$, $\un{\ada}{b}=
1-\epsilon$, $\un{\adb}{a} = 0.5$, $\un{\adb}{b}=0.5$.  Since
$\lambda_1 = 1-\epsilon, \lambda_2 = 0.5$, the unigram language models from which the sponsored answers are sampled are:
$\un{\SCa}{a} = \lambda_1 \un{\pure}{a} + (1-\lambda_1) \un{\ada}{a} = (1-\epsilon) (1-\epsilon) + \epsilon \cdot \epsilon = (1-\epsilon)^2+\epsilon^2$.

$\un{\SCa}{b} = \lambda_1 \un{\pure}{b} + (1-\lambda_1) \un{\ada}{b} = (1-\epsilon) \epsilon + \epsilon (1-\epsilon) = 2\epsilon(1-\epsilon)$.

$\un{\SCb}{a} = \lambda_2 \un{\pure}{a} + (1-\lambda_2) \un{\adb}{a} = 0.5 \cdot (1-\epsilon) + 0.5 \cdot 0.5 = 0.75 - \frac{\epsilon}{2}$.

$\un{\SCb}{b} = \lambda_2 \un{\pure}{b} + (1-\lambda_2) \un{\adb}{b} = 0.5 \cdot \epsilon+ 0.5 \cdot 0.5 = \frac{\epsilon}{2} + 0.25$.

We compute $\valuea, \valueb, \userUtil{\SCa}{\question}, \userUtil{\SCb}{\question}$ using again the linearity of $CE$ in its left argument:

$\valuea = \simFn{\SCa}{\ada} = 2A -CE(\SCa||\ada) -CE(\ada||\SCa) = 2A -(1-\epsilon) CE(\pure||\ada) - \epsilon CE(\ada||\ada) + \epsilon log((1-\epsilon)^2 + \epsilon^2) + (1-\epsilon)log(2\epsilon(1-\epsilon)) = 2A -CE(\pure||\ada) + log(2\epsilon)$.

$\valueb = \simFn{\SCb}{\adb} = 2A -CE(\SCb||\adb) -CE(\adb||\SCb) = 2A - 0.5 CE(\pure||\adb) - 0.5 CE(\adb||\adb) + 0.5 log(\frac{3}{4}-\frac{\epsilon}{2}) + 0.5 log(\frac{1}{4}+\frac{\epsilon}{2}) = 2A - 0.5 CE(\pure||\adb) - 0.5 CE(\adb||\adb) + 0.5 log(\frac{3}{4}) + 0.5 log(\frac{1}{4})$.

$\userUtil{\SCa}{\question} = \simFn{\SCa}{\pure} = 2A - CE(\SCa||\pure) - CE(\pure||\SCa) = 2A - (1-\epsilon) CE(\pure||\pure) - \epsilon CE(\ada||\pure) + (1-\epsilon) log((1-\epsilon)^2 + \epsilon^2) + \epsilon log(2\epsilon(1-\epsilon)) = 2A - CE(\pure||\pure)$.

$\userUtil{\SCb}{\question} = \simFn{\SCb}{\pure} = 2A - CE(\SCb||\pure) -CE(\pure||\SCb) = 2A - 0.5 CE(\pure||\pure) - 0.5 CE(\adb||\pure) + (1-\epsilon)log(\frac{3}{4}-\frac{\epsilon}{2}) + \epsilon log(\frac{1}{4}+\frac{\epsilon}{2}) =  2A - 0.5 CE(\pure||\pure) - 0.5 CE(\adb||\pure) + log(\frac{3}{4})$.

We compute all the following $CE(\cdot||\cdot)$ terms:

$CE(\ada||\ada), CE(\ada||\pure), CE(\pure||\ada) CE(\pure||\adb),\\ CE(\adb||\pure), CE(\pure, \pure), CE(\adb, \adb)$:

$CE(\ada||\ada) = -\epsilon log(\epsilon) -(1-\epsilon)log(1-\epsilon)=0$.

$CE(\ada||\pure) = -\epsilon log(1-\epsilon) -(1-\epsilon)log(\epsilon) = -(1-\epsilon)log(\epsilon) = -log(\epsilon)$.

$CE(\pure||\ada) = -(1-\epsilon)log(\epsilon) -\epsilon log(1-\epsilon) = -(1-\epsilon)log(\epsilon) = -log(\epsilon)$.

$CE(\pure||\adb) = -(1-\epsilon)log(0.5) -\epsilon log(0.5) = -log(0.5)$.

$CE(\adb||\pure) = -0.5log(1-\epsilon) -0.5log(\epsilon) = -0.5log(\epsilon)$.

$CE(\pure||\pure) = -(1-\epsilon)log(1-\epsilon) -\epsilon log(\epsilon) = 0$.

$CE(\adb||\adb) = -0.5log(0.5)-0.5log(0.5)=-log(0.5)$.

We show that the user utility given the sponsored answer of the first
advertiser is larger than the user utility given the sponsored answer
of the second advertiser:

$\userUtil{\SCa}{\question}
> \userUtil{\SCb}{\question} \Leftrightarrow 2A - (1-\epsilon)
CE(\pure||\pure) - \epsilon CE(\ada||\pure) + (1-\epsilon)
log((1-\epsilon)^2 + \epsilon^2) + \epsilon log(2\epsilon(1-\epsilon))
> 2A - 0.5 CE(d_0||d_0) - 0.5 CE(d_2||d_0) +
(1-\epsilon)log(\frac{3}{4}-\frac{\epsilon}{2}) + \epsilon
log(\frac{1}{4}+\frac{\epsilon}{2}) \Leftrightarrow 2A -
CE(\pure||\pure) > 2A - 0.5 CE(d_0||d_0) - 0.5 CE(d_2||d_0) +
log(\frac{3}{4}) \Leftrightarrow 0 > 0.25log(\epsilon) +
log(\frac{3}{4}) \Leftrightarrow \epsilon < \frac{4^4}{3^4}$.

We compute the prospect platform value for both advertisers in order to determine who is the winner:

$\VP(\SCa, \question, \valuea) = \valuea + \userUtil{\SCa}{\question} = 2A -(1-\epsilon) CE(\pure||\ada) - \epsilon CE(\ada||\ada) + \epsilon log((1-\epsilon)^2 + \epsilon^2) + (1-\epsilon)log(2\epsilon(1-\epsilon)) + 2A - (1-\epsilon) CE(\pure||\pure) - \epsilon CE(\ada||\pure) + (1-\epsilon) log((1-\epsilon)^2 + \epsilon^2) + \epsilon log(2\epsilon(1-\epsilon))  = 2A - CE(\pure||\ada) + log(2\epsilon) + 2A - CE(\pure||\pure) = 4A + log(\epsilon) + log(2\epsilon)$.

$\VP(\SCb, \question, \valueb) = \valueb + \userUtil{\SCb}{\question} =  2A - 0.5 CE(\pure||\adb) - 0.5 CE(\adb||\adb) + 0.5 log(\frac{3}{4}-\frac{\epsilon}{2}) + 0.5 log(\frac{1}{4}+\frac{\epsilon}{2}) + 2A - 0.5 CE(\pure||\pure) - 0.5 CE(\adb||\pure) + (1-\epsilon)log(\frac{3}{4}-\frac{\epsilon}{2}) + \epsilon log(\frac{1}{4}+\frac{\epsilon}{2}) = 4A +0.5log(0.5) + 0.5log(0.5) + 0.5log(\frac{3}{4}) + 0.5log(\frac{1}{4}) + 0.25log(\epsilon) + log(\frac{3}{4}) = 4A + log(\frac{3\sqrt{3}}{32}) + 0.25log(\epsilon)$.

The winner is again advertiser $i=2$, since she maximizes the prospect platform value: $\valuea + \userUtil{\SCa}{\question} < \valueb + \userUtil{\SCb}{\question} \Leftrightarrow 4A +log(\epsilon) + log(2\epsilon) < 4A + log(\frac{3\sqrt{3}}{32}) + 0.25log(\epsilon) \Leftrightarrow 2log(\epsilon) + log(2) < log(\frac{3\sqrt{3}}{32}) + 0.25log(\epsilon)  \Leftrightarrow 1.75log(\epsilon) < log(\frac{3\sqrt{3}}{32}) - log(2) \Leftrightarrow \epsilon < 0.23.$